\documentclass{sig-alternate}
\usepackage{mathptmx}
\usepackage{graphicx}
\usepackage[ruled, vlined]{algorithm2e}
\usepackage{epsfig}
\usepackage{amsfonts}
\usepackage{amssymb}
\usepackage{epstopdf}

\usepackage{balance}  

\newtheorem{mylemma}{Lemma}
\begin{document}

\title{Scalable Data Cube Analysis over Big Data}

\author{
   \begin{tabular*}{0.7\textwidth}{@{\extracolsep{\fill}}ccc}
       Zhengkui Wang$^\ddag$ &
       Yan Chu $^\ddag$ $^*$ &  
        Kian-Lee Tan$^\ddag$ 
         \\
         \emph{wangzhengkui@nus.edu.sg} & \emph{chuyan@hrbeu.edu.cn} & \emph{tankl@comp.nus.edu.sg} \\
          Divyakant Agrawal$^\S$ &
        Amr EI Abbadi$^\S$ &
        Xiaolong Xu$^\#$ \\
        \emph{agrawal@comp.ucsb.edu} & \emph{amr@comp.ucsb.edu} & \emph{qlxucn@gmail.com}
   \end{tabular*}\\
   \\
   \begin{tabular*}{0.9\textwidth}{@{\extracolsep{\fill}}ccc}
      $^\ddag$ National University of Singapore &  $^*$ Harbin Engineering University
             \\  $^\S$ University of California at Santa Barbara & $^\#$ AppFoilo Inc. 
        \end{tabular*}
}
\maketitle

\begin{abstract}

Data cubes are widely used as a powerful tool to provide multi-dimensional views in data warehousing and On-Line Analytical Processing (OLAP). However, with increasing data sizes, it is becoming computationally
expensive to perform data cube analysis. The problem is exacerbated by
the demand of supporting more complicated aggregate functions (e.g. CORRELATION, Statistical Analysis) as well as supporting frequent view updates in data cubes. This calls for new scalable and efficient data cube analysis systems.
In this paper, we introduce {\tt HaCube},  an extension of MapReduce, designed for efficient parallel data cube analysis on large-scale data by taking advantages from both MapReduce (in terms of scalability) and parallel DBMS (in terms of efficiency).
We also provide a general data cube materialization algorithm which is able to facilitate the features in MapReduce-like systems towards an efficient data cube computation. Furthermore, we demonstrate how {\tt HaCube} supports view maintenance through either incremental computation (e.g. used for SUM or COUNT) or recomputation  (e.g. used for MEDIAN or CORRELATION).
We implement {\tt HaCube} by extending Hadoop and evaluate it based on the TPC-D benchmark over billions of tuples on a cluster with over 320 cores. The experimental results demonstrate the efficiency, scalability and practicality of {\tt HaCube} for cube analysis over a large amount of data in a distributed environment.

\end{abstract}

\section{Introduction}
In many industries, such as sales, manufacturing, transportation and finance,
there is a need to make decisions based on aggregation of data over
multiple dimensions. {\tt Data cubes} \cite{gray:datacube} are one
such critical technology that has been widely used
in data warehousing and On-Line Analytical Processing (OLAP)
for data analysis in support of decision making.

In OLAP, the attributes are classified into \textbf{dimensions} (the grouping attributes) and \textbf{measures} (the attributes which are aggregated) \cite{gray:datacube}. Given $n$ dimensions, there are a total of $2^{n}$ {\tt cuboids}, each of which captures the aggregated data over
one combination of dimensions. To speed up query processing, these cuboids are typically stored into a database as views. The problem of \textbf{data cube materialization} is to efficiently compute all the views ($\mathbb{V}$) based on the data ($\mathbb{D}$).  Fig. \ref{fig:lattice} shows all the cuboids represented as a cube lattice with 4 dimensions $A$, $B$, $C$ and $D$.

In many append-only applications (no UPDATE and DELETE operations), the new data ($\Delta \mathbb{D}$) will be incrementally INSERTed or APPENDed to the data warehouse for view update. For instance, the logs in many applications (like the social media or stocks) are incrementally generated/updated. There is a need to update the views in a manner of one-batch-per-hour/day.  The problem of \textbf{view maintenance} is to efficiently calculate the latest views while $\Delta \mathbb{D}$ are produced.

Both data cube materialization and view maintenance are computationally expensive, and have
received considerable attention in the literature \cite{xin:tkde}\cite{zhao:arraycube}\cite{lee:incremain}\cite{palpanas:increnondis}.
However, existing techniques can no longer meet the demands of
today's workloads. On the one hand, the amount of data is increasing
at a rate that existing techniques
(developed for a single server or a small
number of machines) are unable to offer acceptable performance.
On the other hand, more complex aggregate functions
(like complex statistical operations) are required
to support complex data mining and statistical analysis tasks.
Thus, this calls for new scalable systems to efficiently support data cube
analysis over a large amount of data.

Meanwhile, MapReduce (MR) \cite{dean:mapreduce} has emerged as a
powerful computation paradigm for parallel data processing on large-scale clusters.
Its high scalability and append-only features have made it
a potential target platform for data cube analysis in append-only applications.
Therefore, exploiting MR for data cube analysis has become an interesting research topic.
However, deploying an efficient data cube analysis using MR is non-trivial. A naive implementation of cube materialization and view maintenance over MR can result in high overheads.

We first summarize the main challenges for cube analysis on large-scale data for developing an efficient cube analysis system.
\begin{itemize}
  \item {
Given $n$ dimensions in one relation, there are $2^{n}$ cuboids that need
to be computed in the cube materialization. An efficient parallel
algorithm to materialize the cube faces two challenges:
(a) Given that some of the
cuboids share common dimensions,  is it possible to batch
these cuboids to exploit some common processing? (b) Assuming
we are able to create batches of cuboids, how can we allocate
these batches or resources so that the load across the processing
nodes is balanced? }

  \item {
View maintenance in a distributed environment introduces significant overheads, as large amounts of data (either the old materialized data or the old base data) needs to be read, shuffled and written among the processing nodes and
distributed file system (DFS). Moreover, more and more applications request to perform view updates more frequently than before, shifting from one-update-per-week/month to almost one-update-per-day even per hour. It is thus critical
to develop efficient view maintenance methods for frequent view updates, even realtime updates. }
\end{itemize}

\begin{figure}[t]
    \begin{minipage}{0.49\textwidth}
    \centerline{\psfig{figure=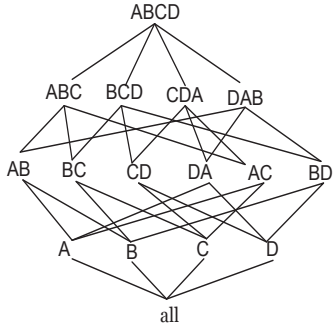,width=50mm,height=45mm} }
    \caption{A cube lattice with 4 dimensions: A, B, C and D}
    \label{fig:lattice}
    \end{minipage}
\end{figure}
Therefore, in this paper, we are motivated to explore the techniques of developing new scalable data cube analysis systems by leveraging the MR-like paradigm, as well as to develop new techniques for efficient data cube analysis to broaden the application of data cubes primarily for append-only environments.
Our main contributions are as follows:

1. \emph{New system design and implementation:} We present \emph{{\tt HaCube}}, an extension of MR, for large-scale data cube analysis. \emph{{\tt HaCube}} tries to integrate the good features from both MR and parallel DBMS. \emph{{\tt HaCube}} extends MR to better support data cube analysis by integrating new features, e.g. a new local store for data reuse among jobs, a layer with user-friendly interfaces and a new computation paradigm MMRR ({\tt MAP}-{\tt MERGE}-{\tt REDUCE}-{\tt REFRESH}). \emph{{\tt HaCube}} illustrates one way to develop a scalable and efficient decision making system, such that cube analysis can be utilized in more applications.


2. \emph{A General Cubing Algorithm:} We provide a general and efficient data cubing algorithm, {\tt CubeGen}, which is able to complete the entire cube lattice using one MR job. We show how cuboids can be batched together to minimize the read/shuffle overhead and salvage partial work done. On the basis of batch processing principle, {\tt CubeGen} further leverages the ordering property of the reducer input provided by the MR-like framework for an efficient materialization. In addition, we propose one load balancing approach, $LBCCC$ to calculate the number of computation resources (reducers) for each batch, such that the load to each reducer is balanced.

 3. {\emph{Efficient View Maintenance Mechanisms:} We demonstrate how views can be efficiently updated under {\tt HaCube} through either recomputation (e.g. used for MEDIAN or CORRELATION) or incremental computation (e.g. used for SUM or COUNT). To the best of our knowledge, this is the first work to address data cube view maintenance in MR-like systems.

%

4. \emph{Experimental Study:} We evaluate {\tt HaCube} based on the TPC-D benchmark with more than two billions tuples. The experimental results show that {\tt HaCube} has significant performance improvement over MR.

The rest of the paper is organized as follows. Section \ref{sec:background} provides the preliminary knowledge of MR computation paradigm. In Section \ref{sec:overview}, we provide an overview of {\tt HaCube}.
Sections \ref{sec:viewmaterialization} and \ref{sec:incrementalvm} present
our proposed cube materialization and view maintenance approaches.
We discuss the fault tolerance and other issues in Section \ref{sec:faulttolerance}, and report our experimental results in Section \ref{sec:experiment}.
In Section \ref{sec:relatedwork} and Section \ref{sec:conclusion}, we review some related works and conclude the paper.

\section{Preliminaries}
\label{sec:background}

In this section, we provide a brief overview of the MR computation paradigm.
MapReduce has emerged as a powerful parallel computation paradigm \cite{dean:mapreduce}. It has been widely used by various applications such as scientific data processing \cite{wang:cosac}\cite{wang:ceo}\cite{wang:eceo}, text processing \cite{lin:text}\cite{sadasivam:sequence}, data mining \cite{chang:media} \cite{zhao:kmeans} and machine learning \cite{chu:machinelearning} \cite{gillick:ml} and so on. MapReduce has several advantages which make it an attractive platform for large-scale data processing, such as its high scalability (scalability of thousands of machines), good fault tolerance (automatic failure recovery by the framework), ease-of-programming (simple programming logic) and high integration with cloud(availability to every user and low expense in a pay-as-you-go cloud model).

 Under the MapReduce framework, the system architecture of a cluster consists of two kinds of nodes, namely, the NameNode and DataNodes. The NameNode works as a master of the file system, and is responsible for splitting data into blocks and distributing the blocks to the data nodes (DataNodes) with replication for fault tolerance. A JobTracker running on the NameNode keeps track of the job information, job execution and fault tolerance of jobs executing in the cluster. A job may be split into multiple tasks, each of which is assigned to be processed at a DataNode.

The DataNode is responsible for storing the data blocks assigned by the NameNode. A TaskTracker running on the DataNode is responsible for the task execution and communicating with the JobTracker.

\begin{figure*}[htb]
    \centerline{\psfig{figure=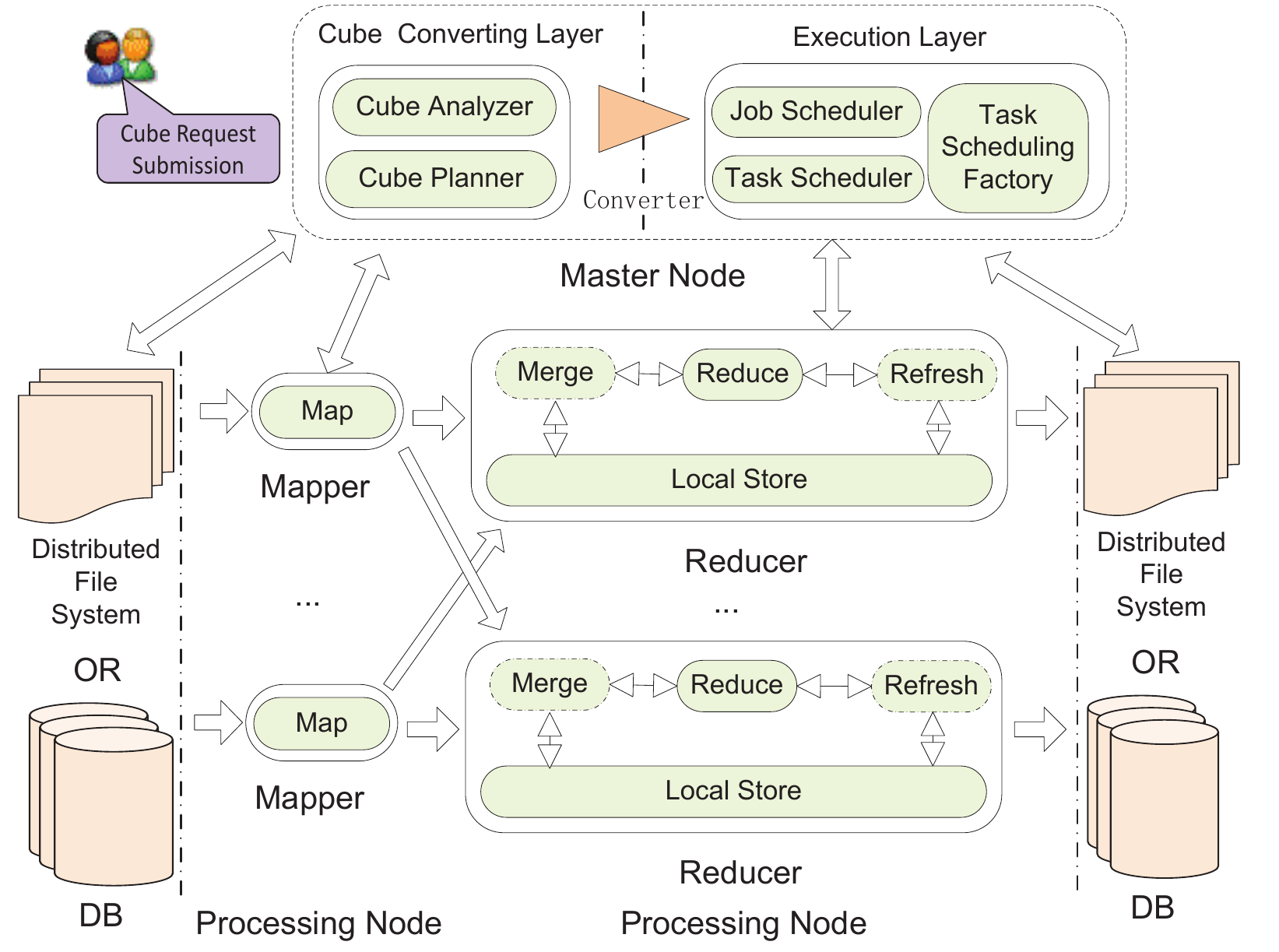,width=130mm,height=80mm} }
    \caption{HaCube Architecture}
    \label{fig:architecture}
\end{figure*}

The computation of MR follows a fixed model with a {\tt map}
phase, followed by a {\tt reduce} phase. Users can set their own computation logic by writing the {\tt map} and {\tt reduce} functions in their applications.

\textbf{\emph{Map Phase:}} The MR library is
responsible for splitting the data into chunks from the distributed system (DFS) and distributing each
chunk to a processing unit (called {\tt mapper}) on different
nodes. The {\tt map} function is used to process
$(key,value)$ pairs $(k1,v1)$ read from data chunks and, after applying the {\tt map} function, then emits a new set of intermediate $(k2,v2)$ pairs.

The MR library sorts and partitions all the intermediate data into $r$ partitions based on the {\tt partitioning} function in each mapper, where $r$ is the number of processing units (called {\tt reducers}) for further computation. The partitions with the same partition number are shuffled to the same reducer. We note that MR randomly chooses the free reducer to process the partitions.

  \textbf{\emph{Reduce Phase:}} The MR library merge-sorts all the $(key,$ $value)$ pairs based on the key first.  Then, the globally sorted data are supplied to the {\tt reduce} function iteratively. After the reduce process, the reducer emits new $(k3,v3)$ pairs to the DFS.

When one job finishes, all intermediate data are removed from the {\tt mappers} and {\tt reducers}. If another job wants to use the same data, it has to reload the data from the DFS again.

\section{HaCube: The Big Picture}
\label{sec:overview}

\subsection{Architecture}

Figure \ref{fig:architecture} gives an overview of the basic architecture of \emph{{\tt HaCube}}. We implement {\tt HaCube} by modifying Hadoop which is an open source equivalent implementation of MR \cite{hadoop}. Similar to MR, all the nodes in the cluster are divided into two different types of function nodes, including the master and processing nodes. The master node is the controller of the whole system and the processing nodes are used for storage as well as computation.

\textbf{\emph{Master Node:}} The master node consists of two
functional layers:

1. The \textbf{cube converting layer} contains two main components: {\tt Cube Analyzer} and {\tt Cube Planner}. The cube analyzer is designed to accept the user request of data cube analysis, analyze the cube, such as figuring
out the cube id (the identifier of the cube analysis application),
analysis model (materialization or view update), measure operators (aggregation function),
and input and output paths etc.

The cube planner is developed to convert the cube analysis request into an execution job (either a materialization job or a view update job). The execution job is divided into multiple tasks each of which handles part of
the cuboid calculation. The cube planner consists of several functional components such as the execution plan generator (combine the cuboids into batches to reduce the overhead),
and load balancer
(assign the right number of computation resources for each batch).

2. The \textbf{execution layer} is responsible for managing the execution
of jobs passed from the
cube converting layer. It has three main components: {\tt job scheduler}, {\tt task scheduler} and {\tt task} {\tt scheduling} {\tt factory}. We use the same job scheduler as in Hadoop which is used to schedule different jobs from different users. In addition, we add a task scheduling factory which is used to record the task scheduling information of a job which can be reused in other jobs.
Furthermore, we develop a new task scheduler to schedule the tasks in terms of the scheduling history stored in the task scheduling factory rather than the random scheduler used in MR.


\textbf{\emph{Processing Node:}}
A processing node is responsible for the task execution assigned from the master node. Similar to MR, each processing node contains one or more processing units each of which can either be a {\tt mapper} or a {\tt reducer}. Each processing node has a {\tt TaskTracker} which is in charge of communicating with the master node through heartbeats, reporting its status, receiving the task, reporting the task execution progress and so on. Unlike MR, there is a {\tt Local Store} built at
each processing node running {\tt reducers}. The local store is developed to cache useful data of a job in the local file system of the reducer node. It is a persistent storage in the local file system and will not be deleted after a job execution.
In this way, tasks (possibly from other jobs) assigned to the
same reducer node are able to
access the local store directly from the local file system.

\subsection{Computation Paradigm}

{\tt HaCube} inherits some features from MR,
such as data read/proce- ss/write format of (key, value) pairs, sorting
all the intermediate data and so on. However, it further enhances MR to support
 a new computational paradigm.
{\tt HaCube} adds two optional phases
- a {\tt Merge} phase and a {\tt Refresh} phase
before and after the {\tt Reduce} phase - to support the
MAP-MERGE-REDUCE-REFRESH (MMRR) paradigm as shown
in Fig. \ref{fig:architecture}.

The {\tt Merge} phase has two functionalities. First, it is used to cache the data from the reduce input to the local store. Second, it is developed to sort and merge the partitions from mappers with the cached data in the local store. The {\tt Refresh} phase is developed to perform further computations based on the reduce output data. Its functionalities include caching the reduce output data to the local store and refreshing the reduce output data with the cached data in the local store. These two additional phases are intended to fit different application requirements for efficient execution support.

As mentioned, these two phases are optional for the jobs.
Users can choose to use the original MR computation or
MMRR computation. More details can be found in Section \ref{sec:incrementalvm} about how MMRR benefits the data cube view maintenance.

\section{Cube Materialization}
\label{sec:viewmaterialization}

In this section, we provide our proposed data cubing algorithm, {\tt CubeGen}, under the MR-like systems. We first present some principles of sharing computation through cuboid batching followed by a batched generator. We then introduce the load balancing strategy followed by the detail implementation of {\tt CubeGen}.
For simplicity, we assume that we are materializing the complete cube. Note that our techniques can be easily generalized to compute a partial cube
(compute only selective cuboids). We also omit the cuboid ``all" from the lattice. This special cuboid can be easily handled through an independent processing unit.

\subsection{Cuboid Computation Sharing}

To build the cube, computing each cuboid independently
is clearly inefficient. A more efficient solution, which
we advocate, is to combine cuboids into batches so that
intermediate data and computation can be shared and salvaged.

We provide the following lemma as a formal basis for combining and batching the cuboids computation under MR-like systems.

\begin{mylemma}
\label{lemma:batch}
Let $A$ and $B$ be a set of dimensions such that
$A \bigcap B = \emptyset$.
In MR-like systems, given cuboids $A$ and $AB$,
$A$ can be combined and processed together with $AB$,
once $AB$ is set of the key and is partitioned by $A$ in one MR job.
$A$ is referred to as the \textbf{ancestor} of $AB$ (denoted as $A \prec AB$). Meanwhile, $AB$ is called the \textbf{descendant} of $A$.
Note that the ancestor and descendant require them share the same prefix.
\end{mylemma}
\begin{proof}
Without loss of generality, we assume $A=d_{1},$ ...,  $d_{x}$ and $B=d_{y},$ ..., $d_{y+z}$ where $d_{i}$ is one dimension. When $AB$ is processed,
the output key in the mapper can be set as the group-by dimensions of $AB$ ($d_{1},$ $...,$ $d_{x},$ $d_{y},$ $...,$ $d_{y+z}$) on which the (k,v) pairs are further sorted based on $AB$. Partitioning $AB$
based on $A$ guarantees that all the same group-by cells in $A$ are shuffled to the same reducer where $A$ is also in a sorted order. Therefore, we can process $A$ as well while $AB$ is processed. 
\end{proof}

The above results can be generalized using transitivity:
Since we can combine the processing of the pair of cuboids $\{A, AB \}$
and the pair $\{ AB, ABC \}$, we can also combine the processing of the
three cuboids $\{ A,$ $AB,$ $ABC \}$.
Thus, given one cuboid, all its ancestors can be calculated together as a batch. For instance, in Fig. \ref{fig:lattice}, as $A \prec AB \prec ABC \prec ABCD$,
the cuboids $A$, $AB$, $ABC$
can be processed with $ABCD$. Note that $BC$ cannot be
processed with $ABCD$ because $BC \nprec ABCD$.

Given a batch, the principle to calculate this batch is to
set the \emph{sort dimensions} as the key and partition the (k,v)
pairs based on the \emph{partition dimensions} in the key in the MR-like paradigm.
We formally define these two dimension classes below: \\
\textbf{Definition 1} \textbf{\emph{Sort Dimensions:}} \emph{The dimensions
in cuboid A are called the sort dimensions
if A is the descendant of all other cuboids in one batch.} \\
\textbf{Definition 2} \textbf{\emph{Partition Dimensions:}} \emph{The dimensions in cuboid A are
called the partition dimensions if A is the ancestors of all
other cuboids in one batch.}

For instance, given the batch $ \{A, AB, ABC, ABCD \}$, $ABCD$ and $A$ can be set as the sort and partition dimensions respectively.

The benefits of this approach are: 1) In the reduce phase, the
group-by dimensions are all in sorted order for every cuboid
in the batch, since MR would sort the data before
supplying to the reduce function. This is an efficient way of
cube computation since it obtains sorting for free and no
other extra sorting is needed before aggregation. 2) All the
ancestors do not need to shuffle their own intermediate
data but use their descendant's. This would significantly reduce
the intermediate data size, and thus remove a lot of data
sort/partition/shuffle overheads.

%

To achieve good performance, we need to address two issues.
First, how can we find the minimum number of batches from the $2^{n}$
cuboids? As more cuboids are combined together, the
shuffling overhead incurred for data shuffling will be reduced.
Second, how can we balance the load to assign the right number of
computation resources to each batch?  As different batches may have
different computation complexity and data size, it is not optimal
to evenly assign the computation resources to each batch.
Before providing the detailed
algorithm for {\tt CubeGen}, we first introduce how it solves
the aforementioned challenges by developing a {\tt plan generator}
and a {\tt load balancer}.

\subsection{Plan Generator}

The goal of the plan generator is to generate the minimum number of batches
among the $2^{n}$-$1$ cuboids, excluding $``all''$. The plan generator first divides the
$2^n$-$1$ cuboids into $n$ groups each of which consists of the
cuboids with $i$ dimension attributes. For instance, given the
cube lattice with 4 dimension attributes in Figure \ref{fig:lattice},
it can be divided into 4 groups (from the bottom of the lattice to the top)
as follows: $G_{1}=\{A, B, C, D\}$, $G_{2}=\{AB, BC, CD, DA, AC, BD\}$, $G_{3}=\{ABC, BCD, CDA,$ $DAB\}$, $G_{4}=\{ABCD\}$.

Recall that one cuboid can be batched with all its sub-cuboids.
Thus, we adopt a \emph{greedy approach} to combine one cuboid with
as many of its ancestors as possible. Initially, all the cuboids in each group are marked as available.
Each construction of a batch starts with one available cuboid,
$\alpha$, from the non-empty group with the maximum number of dimensions.
It then searches all the available ancestors of $\alpha$ from other groups
that can be batched together. For instance,
the first batch construction starts with $ABCD$ in the example above (Since
$ABCD$ has 4 dimensions, it is the one with maximum number of dimensions).
Note that since cuboid $\alpha$ has different permutations
(e.g. \emph{ABCD} can also be permuted as $ABDC$, $ACBD$, $BCDA$, $CDAB$,
$DABC$ etc.), the algorithm enumerates all permutations and
the one with the maximum number of available ancestors will be chosen.
Once one batch is constructed, all the cuboids in this batch are
deleted from the search space and become unavailable. Similarly, the next batch
construction is conducted among the remaining available cuboids. The construction finishes when there are no available cuboids left.


The approach we adopt to generate the batches is similar to the one proposed in \cite{lee:incremain}. Lee et al. provide an extensive proof that the algorithm is able to generate $C_{n}^{\lceil \frac{n}{2}  \rceil}$ batches which is the minimum number. Recall that there are $C_{n}^{\lceil \frac{n}{2}  \rceil}$ cuboids in
group $G_{\lceil \frac{n}{2}  \rceil}$
and that none
of them can be combined with each other.
So there are at least $C_{n}^{\lceil \frac{n}{2}  \rceil}$ batches.
Interested
readers are referred to \cite{lee:incremain} for more details.

\begin{figure}[t]
    \centerline{\psfig{figure=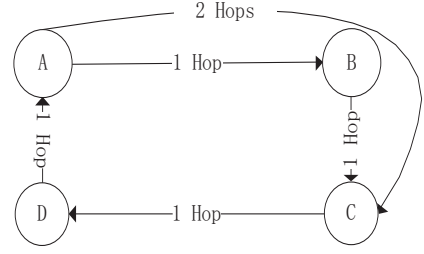,width=55mm,height=35mm} }
    \caption{A directed graph of expressing 4 dimensions $A$, $B$, $C$ and $D$}
    \label{fig:orderedlattice}
\end{figure}

To improve the efficiency of batch construction,
two optimizations are adopted to reduce the search space.
\begin{itemize}
  \item First, recall that for a set of dimensions, we need to compute a batch
for each permutation. This, however, may not be necessary.
In fact, when all the sub-cuboids of a particular permutation are available,
we know that we have found a permutation with the maximum number of
sub-cuboids. Therefore, as soon as we encounter such a permutation,
we do not need to continue the search for this set of dimensions.
  \item Second, we organize all the dimensions as one directed graph such that one dimension points to another and we refer to the distance between two adjacent dimensions as one hop. For instance, given 4 dimensions $A$, $B$, $C$ and $D$, they can be expressed as a directed graph such that $A$, $B$, $C$ and $D$ point to $B$, $C$, $D$ and $A$ respectively as shown in Figure \ref{fig:orderedlattice} (b). During permutation enumeration, changing from $A$ to $B$ or $C$ is referred as moving one hop or two hops from $A$.

      To find the permutation of a cuboid $\alpha$ with the maximum number of sub$\_$cuboids, the enumeration starts from the permutation that is obtained by moving the equivalent number of hops for each dimension of the unavailable cuboid in the same group.
      This is to guarantee that, most likely, the first search permutation is the one we need to reduce the search space. For instance, assume that the first batch is generated as $ABCD$, $ABC$, $AB$ and $A$. Then the next initial permutation for a new batch is $BCD$ which is computed through moving one hop for each dimension in the unavailable cuboid $ABC$. It is clear that the new batch can be generated with $BCD$ (since all its sub$\_$cuboids are available) and there is no need to search other permutation of $BCD$.

      In the same way, the batches of $CDA$ and $DAB$ will be generated.  These two optimizations speed up the batch construction.

\end{itemize}

Figure \ref{fig:orderedlattice} shows an example of the generated batches using the dotted lines.


\begin{figure}[t]
     \centerline{\psfig{figure=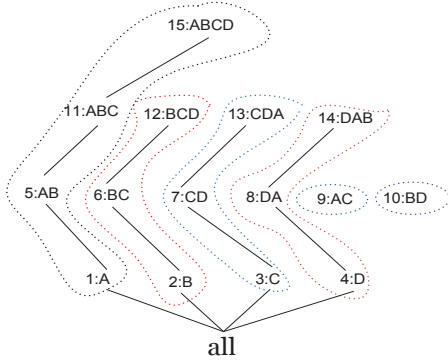,width=60mm,height=50mm} }
    \caption{The numbered cube lattice with execution batches}
    \label{fig:orderedlattice}
\end{figure}


\subsection{Load Balancer}

Given a set of batches from the plan generator, the load balancer is used to assign the right number of computation resources to each batch to balance the load.

We argue that existing works (like \cite{ng:iceberg}) that balance the
batches by evenly assigning the computation resources may not always be
a good choice.
\begin{itemize}
  \item First, it requires users to provide very specific information
about the application data to be able to estimate a
model to find the balanced batches.
  \item Second, the cuboids may not be
combined into balanced batches.
  \item  Third, in a MR-like system,
it is hard to make a precise cost estimate of each batch.
For instance, the total cost of each batch includes the following main parts:
the data shuffling cost (shuffling the intermediate data from mappers to
reducers), the sorting cost (all the intermediate data are sorted),
data processing cost (applying the measure function to each cuboid
in a batch) and data writing cost (writing the views to the file system).
It is hard to estimate each of these component costs and even
harder to evaluate the total cost of each batch (as this requires
setting the appropriate weights when combining these components).
\end{itemize}

Therefore,
the load balancing among different batches becomes a very tricky and
challenging problem in MR-like systems.

In this paper, we propose a novel load balancing scheme
$LBCCC$ (short for Load Balancing via Computation Complexity Comparison)
to assign the right number of computation resources to each batch. Intuitively,
$LBCCC$ adopts a profiling based approach where a learning job (we refer as $CCC$ - Computation Complexity Comparison job ) is first conducted on a
small test dataset to evaluate the computation overhead relationship between each batch and then generate the number
of reducers for each batch that are proportional to computation overhead for the actual {\tt CubeGen} jobs.
The computation cost relationship is estimated through the execution comparisons when each
batch is provided the same number of computation resources.
The execution time relationship over the same computation
resource indicates the entire batch processing overhead relationship, thus
helps to make an accurate load
balancing decision.

In particular, $LBCCC$ first conducts the $CCC$ job, a cube materialization learning job, on a small test dataset
where each batch is assigned to one reducer. It then records and utilizes the execution time of each batch to
estimate the computation overhead relationship among different batches.

The test data can be obtained either by sampling or produced within a time window provided by users. Note that the sampling can be accomplished during the $CCC$ job, since the MR framework provides
APIs for sampling data directly. Therefore, by default, we use the sampling approach provided by the Hadoop API where one tuple is sampled from every $s$ records. Users can also plug in their own sampling algorithm easily.
The sampling algorithms have been widely studied in the literature. Since it is not our focus
on studying how to choose or design a good sampling algorithm, we shall not discuss it here and more sampling algorithms can be found in  \cite{gemulla:sampling}.

In the $CCC$ job, given a set of base data, the {\tt mapper} conducts sampling on it.
For each sampling tuple, the {\tt mapper}
emits multiple (key,value) pairs each of which is for one batch.
Then the $CCC$ job shuffles the pairs that belong to the same
batch to one particular reducer. Given $b$ batches,
the $CCC$ job uses $b$ reducers each of which is in charge of
processing one batch. The implementation of the $CCC$ algorithm is
similar to the {\tt CubeGen} algorithm provided in Algorithm \ref{algo:cubegen} except for the number of reducers assigned to each batch. The algorithm detail can be found in section \ref{subsec:cubegenimplementation}.

The $CCC$ learning job records the execution time $T_{i}$
for processing batch $B_{i}$.
Based on the execution time recorded, the load balancer generates the resources assignment plan for each
batch. Given $r$ reducers, the right number of
reducers, $R_{i}$ for batch $B_{i}$ can be calculated as follows:
    \begin{center}
    $R_{i}=\frac{T_{i}*r}{\sum_{j=0}^{b-1}{T_{j}}}$
    \end{center}

The load balancer integrates this plan into the {\tt CubeGen} algorithm to balance the load.
The experimental results show that $LBCCC$ is able to balance the load very well.

We note that this evaluation only needs to be done once before the initial cube materialization and is used for subsequent jobs in the same application.
Furthermore, performing the $CCC$ job is cheap as only $b$ reducers are needed.
Note that the load balancer can support different kinds of batching approaches.
Therefore, the $LBCCC$ load balancing scheme is general and effective for different cubing algorithms.
\begin{algorithm} [t] \small
\caption{\textbf{\small CubeGen Algorithm} }
\label{algo:cubegen}
Function: Map(t)
\\ \nl $\#$  t is the tuple value from the raw data
\\ \nl Let $\mathbb{B}$ (resp. $I_{i}$) be the batch set with $B_{0}$, $B_{1}$, ..., $B_{b-1}$ (resp. the identifier of batch $B_{i}$)
\\ \nl \For{each $B_{i}$ in $\mathbb{B}$}
{ \nl k (resp. v) $\Leftarrow$ get sort dimensions (resp. the measure m) in $B_{i}$ from t
\\ \nl $\#$ If there are multiple measures (e.g. $m_{1},m_{2}$), then v $\Leftarrow$ $(m_{1}, m_{2})$
\\ \nl v.append($I_{i}$); emit(k,v);
}

Function: Partitioning($k, v$)
\\ \nl Let $R_{i}$ (resp. $attr$) be the number of reducers (resp. the partition dimensions) for $B_{i}$
\\ \nl $S_{i}$ $\Leftarrow$ $\sum_{j=0}^{i-1}{R_{j}}$
\\ \nl return $S_{i}$ + $hash(attr, R_{i})$;

Function: Reduce/Combine ($k,\{v_{1},v_{2},...,v_{m} \}$)
\\ \nl Let $\mathbb{C}$ (resp. $\mathbb{M}$) be the cuboid set in the batch identifier (resp. the aggregate function)
\\ \nl \For  {$C_{i}$ in $\mathbb{C}$}
{
\nl \If {$C_{i}$ is ready}
{
 \nl  $k''$ (resp. $v''$)$\Leftarrow$ get the group-by dimensions in $C_{i}$ (resp. $ \mathbb{M}(v_{1}, ..., v_{m}, v_{1}^{'}, ..., v_{k}^{'},...) $)
\\ \nl $\#$ Perform multiple aggregate functions e.g. $(\mathbb{M}_{1}, \mathbb{M}_{2})$ here: $v_{1}''$  $\Leftarrow$ $ \mathbb{M}_{1}(v_{1}, ..., v_{m}, v_{1}^{'}, ..., v_{k}^{'},...) $ and $v_{2}''$  $\Leftarrow$ $ \mathbb{M}_{2}(v_{1}, ..., v_{m}, v_{1}^{'}, ..., v_{k}^{'},...) $
\\ \nl emit($k'',v''$);
}
\nl \Else
{
\nl Buffer the measure for aggregation
}
}
\end{algorithm}

\subsection{Implementation of CubeGen}
\label{subsec:cubegenimplementation}

Consider the batch plan $B$ ($B_{0}$, $B_{1}$, ..., $B_{b-1}$) generated from the plan generator and the load balancing plan R=($R_{0}$, $R_{1}$, ..., $R_{b-1}$) where the $R_i$ is the number of reducers assigned for batch $B_i$. Given $B$ and $R$, the proposed {\tt CubeGen} algorithm materializes the entire cube in one job and its pseudo-code is provided in Algorithm \ref{algo:cubegen}.

\textbf{Map Phase:} The base data is split into different chunks each of which is processed by one mapper. {\tt CubeGen} parses each tuple and emits multiple (k,v) pairs each of which is for one batch (lines 3-6).
The sort dimensions in the batch are set as the key and the measure is set as the value.

To distinguish which (k,v) pair is for which batch with which cuboids, we add a batch identifier appended after the value.
The identifier is developed as one Bitmap with $2^{n}$ bits where $n$ is the number of dimensions and each bit corresponds to one cuboid. First, we number all the $2^{n}$ cuboids from $0$ to $2^{n}$-1. Second, if the cuboid is included in one batch, its corresponding bit is set as 1, otherwise 0. For instance, Fig. \ref{fig:orderedlattice} depicts an example of a numbered cube lattice. Assume that $B_{0}$ consists of cuboids $\{A, AB, ABC$ and $ABCD \}$. The identifier for $B_{0}$ is set as $ `10001000$ $00100010 '$.

The partitioning function partitions the pairs to the appropriate partition based on the identifier and the load balancing plan $R$. {\tt CubeGen} first schedules the data into the right range of reducers. Recall that the batch $B_{i}$ is assigned $R_{i}$ reducers. Therefore, the assigned reducers for batch $B_{i}$ are from $\sum_{j=0}^{i-1}{R_{j}}$ to $\sum_{j=0}^{i-1}{R_{j}}$+$R_{i}$-1. Then the (k,v) pairs are hash partitioned among these $R_{i}$ reducers according to the partition dimensions in the key (lines 7-9).

\textbf{Reduce Phase:} In the {\tt Reduce} phase, the MR library sorts all the (k,v) pairs based on the key and passes them to the reduce function. Each reducer obtains its computation tasks (the cuboids in the batch) by parsing the batch identifier in the value. The reduce function
extracts the measure and projects the group-by dimensions for each cuboid in the batch. For the descendant cuboid, the aggregation can be performed directly based on input tuple, since each input tuple is one complete group-by cell. For other cuboids, the measures of the group-by cell are buffered until the cell receives all the measures it needs for aggregation (lines 11-17). We develop multiple file emitters to write different aggregated results to different destinations.

%

Note that if the (k,v) pairs can be pre-aggregated in the {\tt map} phase, users can specify a combine function to conduct a first round aggregation. The combine function is normally similar to the reduce function as shown in lines 10-17, but only aggregates the pairs with the same key. This pre-aggregation is able to reduce the data shuffle size between mappers to reducers.
%

We emphasize that if there are muliple measures (e.g. $m_1$, $m_2$, ..., $m_n$) and multiple aggregate functions ($\mathbb{M}_1$, $\mathbb{M}_2$, .., $\mathbb{M}_m$), they can be processed in the same MR job as shown in the line 5 and 14 in Algorithm \ref{algo:cubegen}. Compared to the naive solution, {\tt CubeGen} minimizes the cube materialization overheads by sharing the data read/shuffle/compu- tation to the maximum, which obtains significant performance improvement as we shall see in Section \ref{sec:experiment}.

\section{View Maintenance}
\label{sec:incrementalvm}

There are two different manners to update the views, namely recomputation and incremental computation. Recomputation computes the latest views by reconstructing the cube based on the entire base data $\mathbb{D}$ and $\Delta \mathbb{D}$. In append-only applications, this manner is normally used for the holistic aggregate functions, e.g. STDDEV, MEDIAN, CORRELATION and REGRESSION \cite{gray:datacube}.

Incremental computation, on the other hand, updates the views using only $\mathbb{V}$ and $ \Delta \mathbb{D}$ in two steps: 1.) In the propagate step,
a delta view $\Delta \mathbb{V}$ is calculated based on the $\Delta \mathbb{D}$. 2.) In the refresh step, the latest view is obtained by merging $\mathbb{V}$ and $\Delta \mathbb{V}$ without visiting $\mathbb{D}$  \cite{mumick:maintenance}.
In append-only applications, this manner is normally used for the distributive and algebraic aggregate functions, e.g. SUM, COUNT, MIN, MAX and AVG \cite{gray:datacube}. Note that the update for these functions can also be conducted through recomputation.

 \begin{algorithm}[t]
\caption{\textbf{ A Refresh Job in MR}}
\label{algo:refreshmr}
Function: Map(t)
\\ \nl $\#$  t is the tuple value from either $\mathbb{V}$ or $\Delta \mathbb{V}$
\\ \nl k (resp. v) $\Leftarrow$ get dimensions (resp. aggregate value) from t;
\\ \nl emit(k,v)

Function: Reduce(k, $\{v_1, v_2\}$)
\\ \nl emit$(k,\mathbb{M}(v_{1},v_{2}))$
\end{algorithm}

\begin{figure*}[t]
 \begin{minipage}[t] {0.5\textwidth}
 \centering
  \epsfig{file=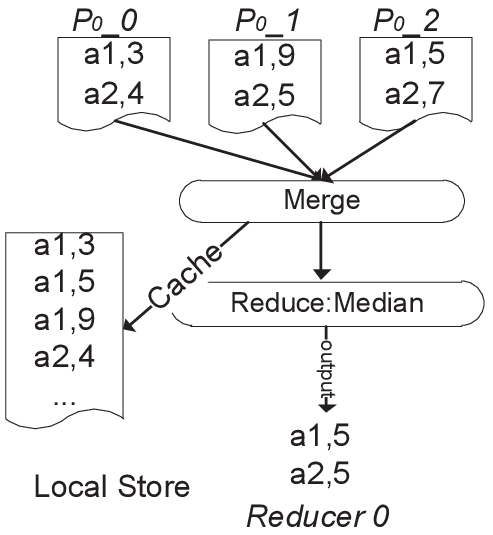,height=7cm, width=5.5cm}\\

      (a) MEDIAN: Cube Materialization (Map-Merge-Reduce)
       \end{minipage}
     \begin{minipage}[t] {0.5\textwidth}
 \centering
  \epsfig{file=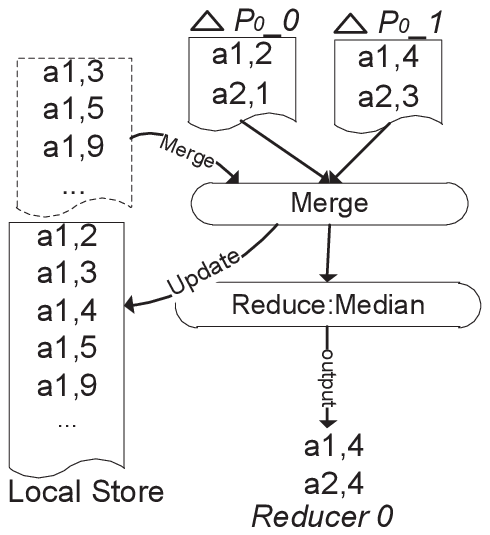,height=7cm, width=5.5cm}\\
      (b) MEDIAN: View Maintenance (Map-Merge-Reduce)
       \end{minipage}
       \caption{Recomputation for MEDIAN in HaCube}
    \label{fig:recomputation}
\end{figure*}

\subsection{Supporting View Maintenance in MR}
\label{subsec:vminmr}

To support recomputation in MR, when $\Delta \mathbb{D}$ is inserted, the latest views can be calculated by issuing one MR job using our {\tt CubeGen} algorithm to reconstruct the cube over $\mathbb{D} \cup \Delta \mathbb{D}$.
The key problem with such a MR-based recomputation view updates is that reconstruction
from scratch in MR is expensive. This is because the base data (which is
large and increases in size at each update) has to be reloaded to the mappers from DFS and shuffled to
the reducers for each view update, which incur significant overheads.

To support incremental computation in MR, the latest views can be calculated by issuing two MR jobs. The first propagate job generates
$\Delta \mathbb{V}$ from $\Delta \mathbb{D}$ using our proposed {\tt CubeGen} algorithm.
The second refresh job merges $\mathbb{V}$ and $\Delta \mathbb{V}$ as shown in Algorithm \ref{algo:refreshmr}.
However, this would incur significant overheads. For instance, the materialized $\Delta \mathbb{V}$ from the
propagate job has to be written back to DFS, reloaded from DFS again and shuffled from mappers to reducers in the refresh job. Likewise, $\mathbb{V}$ has to be reloaded and shuffled around in the refresh job. Therefore, it is highly expensive to support view update operations directly over the traditional MR.

\subsection{HaCube Design Principles}

{\tt HaCube} avoids the aforementioned overheads through storing and reusing the data between different jobs. We extend MR to add a local store in the reducer node which is intended to store useful data of a job in the local file system. Thus, the task shuffled to the same reducer is able to reuse the data already stored there. In this way, the data can be read directly
from the local store (and thus
significantly reducing the overhead that would have been
incurred to read the data from DFS and shuffle them from mappers).

We further extend MR to develop a new task scheduler to guarantee that the same task is assigned to the same reducer node and thus the cached data can be reused among different jobs. Specifically, the task scheduler records the scheduling information by storing a mapping between the data partition number (corresponds to the task) and the TaskTracker (corresponds to the reducer node) and puts it to the task scheduling factory from one job. When a new job is triggered to use the scheduling history from previous jobs, the task scheduler fetches and adopts the scheduling information from the factory to distribute the tasks. The scheduler automatically checks the situation of the over-loaded nodes and re-assigns the task to a nearby processing node.

In addition, two computation phases ({\tt Merge} and {\tt Refresh}) are added to
conduct more computation with the
cached data locally. The {\tt Merge}
phase is added to either cache the intermediate reduce input data in one job or preprocess the data between the
newly arriving data and cached data before the {\tt Reduce} phase. The
{\tt Refresh} phase is added to either cache the reduce output data in one job or postprocess the
reduce output result with the cached data after the {\tt Reduce} phase.

\begin{figure*}[t]
      \begin{minipage}[t] {0.5\textwidth}
 \centering
 \epsfig{file=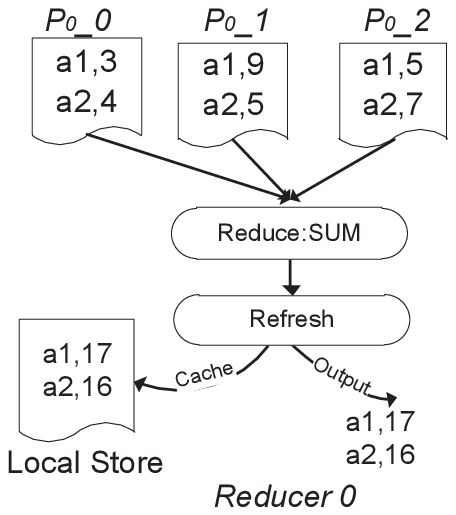,height=7cm, width=5.5cm}\\
      (a) SUM: Cube Materialization (Map-Reduce-Refresh)
       \end{minipage}
     \begin{minipage}[t] {0.5\textwidth}
 \centering
  \epsfig{file=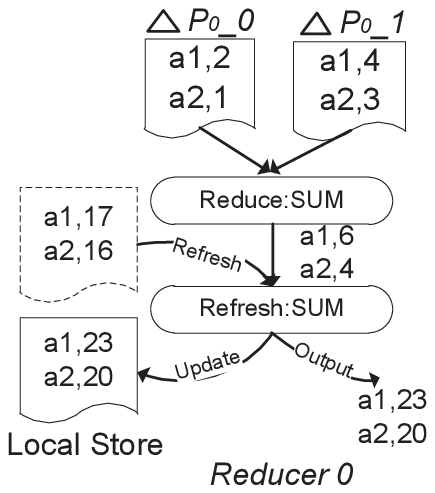,height=7cm, width=6.5cm}\\
      (b) SUM: View Maintenance (Map-Reduce-Refresh)
       \end{minipage}
       \caption{Incremental Computation for SUM in HaCube}
    \label{fig:incomputation}
\end{figure*}
\subsection{Supporting View Maintenance in HaCube}
\label{subsec:nondisvm}

\subsubsection{Recomputation}
The recomputation view update can be efficiently supported in
{\tt HaCube} using a Map-Merge-Reduce (MMR) paradigm.
We demonstrate this procedure through one running example by introducing the cube materialization and update jobs.

In the first cube materialization job, {\tt HaCube} is triggered to cache the intermediate reduce input data to the local store in the {\tt Merge} phase, such that this data can be reused during the view update job. For instance, Fig. \ref{fig:recomputation}(a) shows an example of calculating the cuboid $A$ for MEDIAN.  Assume that reducer 0 is assigned to process cuboid $A$. In this job, each mapper emits one sorted
partition for reducer 0, such as $P_{0}\_0, P_{0}\_1$ and $P_{0}\_2$. Here, each partition is a sequence of
(dimension-value, measure-value) pairs, e.g., (a1, 3), (a2, 4). Recall that once these partitions are shuffled to the reducer 0, it first performs a merge-sort (the same as MR does) to sort all the partitions based on the key in the
 {\tt Merge} phase.  The sorted data is further supplied to the reduce function to calculate the MEDIAN for each group-by cell (e.g. $<a1,5>$ and $<a2,5>$) where this view will be written to DFS.

 Different to MR (which deletes all the intermediate data after one job), since recomputation requires the base data for update, {\tt HaCube}
 caches the sorted reduce input data in the {\tt Merge} phase for subsequent reuse. This caching operation is conducted while the {\tt Reduce} phase finishes, which  guarantees the atomicity of the operation -
if the reduce task fails, the data will not be written to
the local store. Meanwhile, the scheduling information is recorded.

A view update job is launched when $\Delta \mathbb{D}$ is added for view updates. Intuitively, this job conducts a
cube materialization job using the {\tt CubeGen} algorithm based on $\Delta \mathbb{D}$. It
differs from the first materialization job in
the scheduling and the {\tt Merge} phase.
For task scheduling, instead of randomly distributing the tasks
to reducer nodes, it distributes the tasks
according to the scheduling history from the first materialization job to guarantee that the same tasks are processed
at the same reducer.
%
For instance,
the partitions of cuboid $A$
($\Delta P_{0}\_0$ and $\Delta P_{0}\_1$) are scheduled to the same node running reducer 0 as shown in Fig. \ref{fig:recomputation}(b). In the {\tt Merge} phase, since the base data is already cached in the local store, {\tt HaCube} merges the delta partitions with the cached base
data  from the
local store.
Recall that the cached data is the sorted reduce input data
from the previous job, and so it has the same format as the delta
partition. Thus, it can be treated as a local partition
and a global merge-sort is further performed. Then the sorted data will be supplied to the reduce function for recalculation in the {\tt Reduce} phase. When the {\tt Reduce} phase finishes, the local store is updated with both the base and delta data (becoming an updated base
dataset) for further view update use.

Compared to MR, {\tt HaCube} does not need to reload the base data from DFS and shuffle them from mappers to reducers for recomputation. This significantly reduces the data read/shuffle overheads. Another implementation optimization is proposed to minimize the data caching overhead. To cache the data to the local store, it is expensive to push the data to the local store, as this would incur much overhead of moving a large amount of data. Based on the observation that the intermediate sorted data are maintained in temporary files in the local disk in each reducer, {\tt HaCube} simply registers the file locations to the local store other than moving them. Note that the traditional MR would delete these temporary files once one job finishes. As we shall see, the experimental study shows that there is almost no overhead added for caching the data with this optimization.

%
%

\textbf{Incremental Computation}   {\tt HaCube} adopts a
{\tt Map-Reduce-Re- fresh} (MRR) paradigm for incremental computation.
Intuitively, different to MR in the first materialization job, it triggers to invoke a {\tt Refresh} phase after the {\tt Reduce} phase, to cache the view $\mathbb {V}$ to the local store for further reuse.
For instance, Fig. \ref{fig:recomputation} (c) shows an example of calculating cuboid $A$ for SUM in reducer 0. In this job, $\mathbb{V}$ ($<a1,17>$ and $<a2,16>$) is cached to the local store in the {\tt Refresh} phase. In addition, this scheduling information is also recorded.

When $\Delta \mathbb{D}$ is added for view updates, {\tt HaCube} conducts both the propagate and refresh steps in one view update job, as $\mathbb{V}$ is already cached in the reducer node. This view update job in {\tt HaCube} also executes in a MRR paradigm where MR (Map-Reduce) phases obtain $\Delta \mathbb{V}$ based on $\Delta \mathbb{D}$ (propagate step) and the {\tt Refresh} phase merges $\Delta \mathbb{V}$ with $\mathbb{V}$ locally (refresh step). Intuitively, this can be achieved by running the {\tt CubeGen} algorithm on $\Delta \mathbb{D}$ using the same scheduling plan as the previous materialization job. Meanwhile, the cached views in the local store will be updated with the latest ones. For instance, in Fig. \ref{fig:recomputation} (d), the {\tt Reduce} phase calculates the $\Delta \mathbb{V}$ ($<a1,6>$ and $<a2,4>$) based on $\Delta \mathbb{D}$.
In the {\tt Refresh} phase, the updated view ($<a1, 23>$ and $<a2,20>$) is obtained by merging $\Delta \mathbb{V}$ with $\mathbb{V}$ ($<a1,17>$ and $<a2,16>$) cached in the local store.

%

Different to MR, {\tt HaCube} is able to finish the incremental computation in one job where there is no need to reload and shuffle the delta views and old views among DFS and the cluster during the propagate and refresh steps. This provides an efficient view update using the incremental computation by removing much overheads.

\section{Other Issues}
\label{sec:faulttolerance}

\begin{figure*}[htb]
\centering
  \begin{minipage}[t]{.45\textwidth}
    \centering
      \epsfig{file=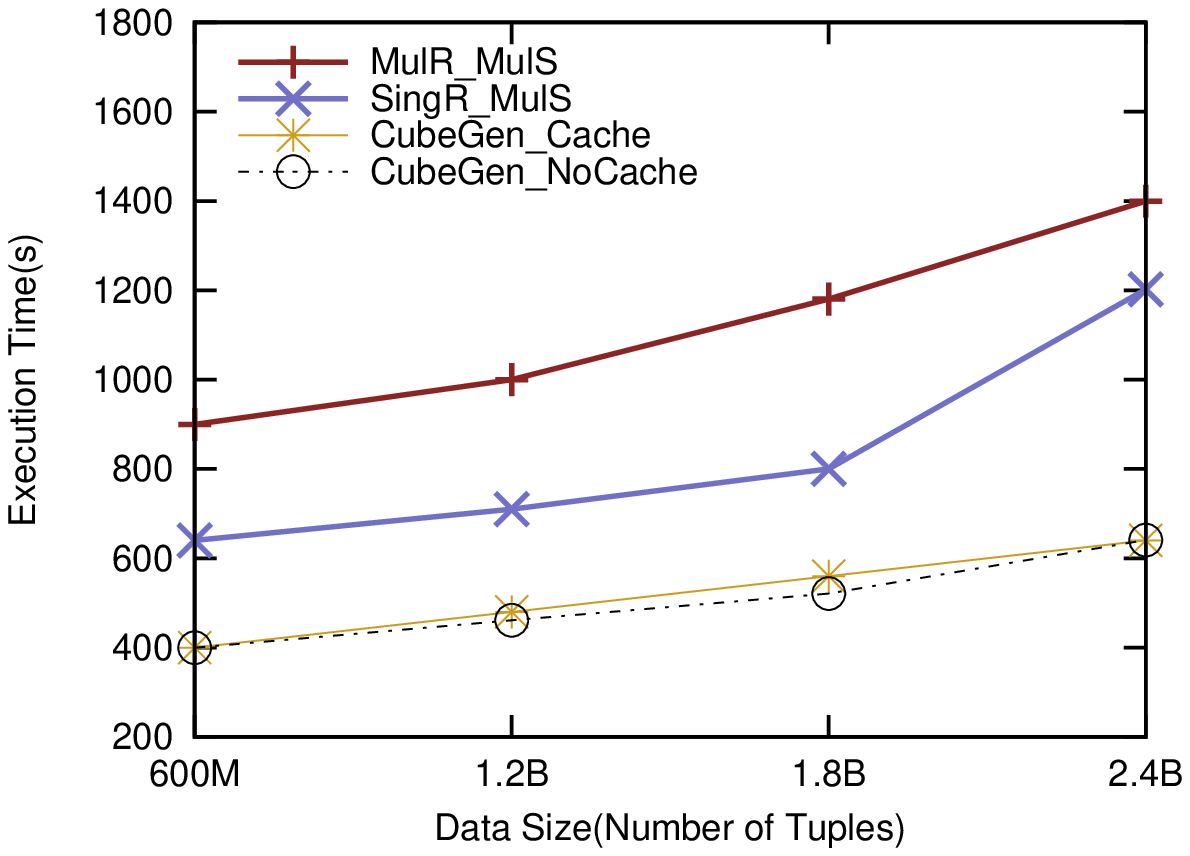,height=7cm, width=4.5cm, angle=270}\\
      (a) MEDIAN
       \end{minipage}
      \begin{minipage}[t]{.45\textwidth}
   \centering
      \epsfig{file=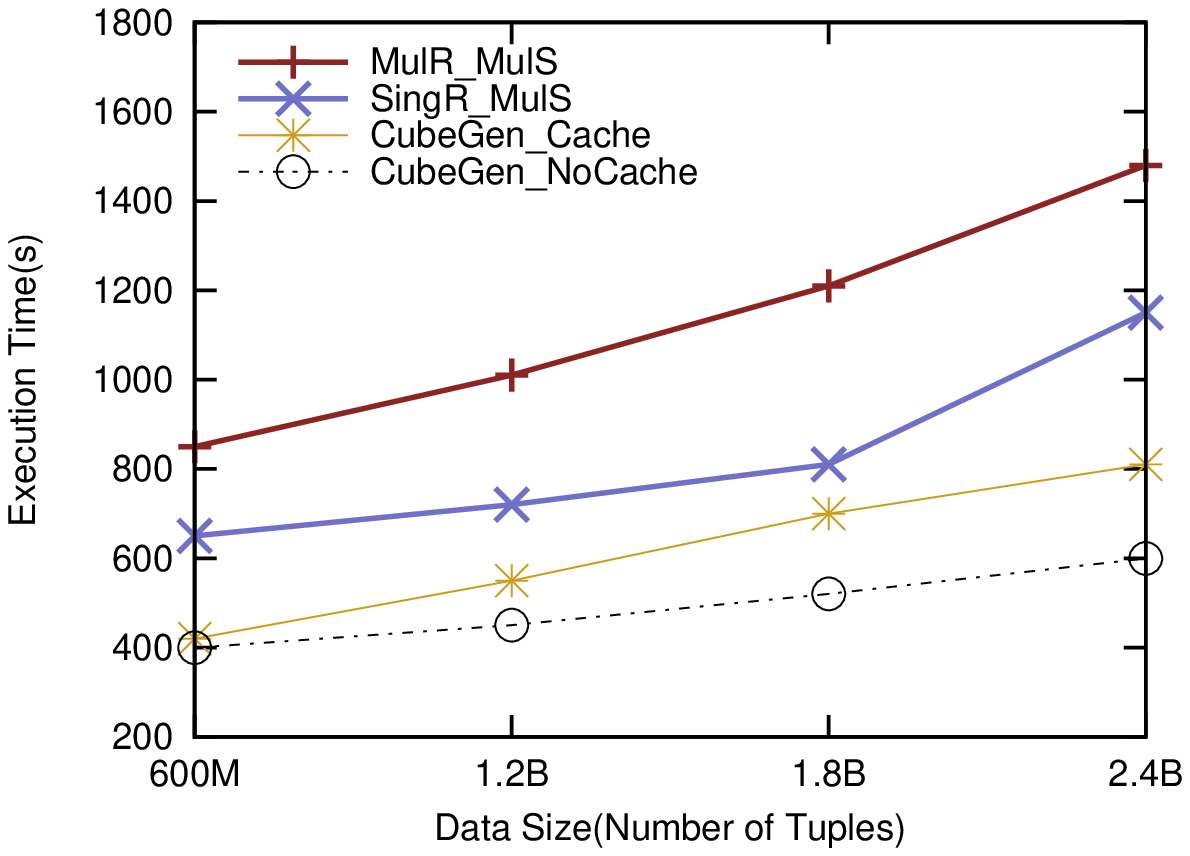,height=7cm, width=4.5cm, angle=270}\\
      (b) SUM
  \end{minipage}

 \caption{\emph{CubeGen} Performance Evaluation for Cube Materialization}
  \label{fig:excubebuilding}
\end{figure*}
\subsection{Fault Tolerance}

Since {\tt HaCube} is built on the MR framework, it preserves the fault-tolerance mechanisms of the MR framework. For instance, the data is replicated in DFS and thus is safe when nodes fail. When a map task fails, the framework schedules the task on another free mapper in the system instead of restarting the whole job.

In addition, we provide an additional fault tolerance strategy to guarantee data availability in the reducer nodes in {\tt HaCube}. The caching mechanism plays an important role in improving the efficiency of data cube analysis. It is important to make sure that the cached data in the reducer node is accessible
when a subsequent job arrives. We handle two kinds of failures in the reducer nodes, including the recoverable and the unrecoverable reducer failures.

\emph{\textbf{Recoverable Failures:}} Recoverable reducer failures
include the task failure and reducer node failure. These failures
can be recovered once the corresponding failed task or node is restarted.
When the task fails in the reducer node, the scheduler kills the task and reschedules it. If the job does not need to use the data in the local store, it will be scheduled to any reducer node. Otherwise, it is scheduled to the same node for data locality. The local store is in a persistent local file system on the reducer node. Thus, after restarting this task, the data is still readable. Similar to the reduce task failure, if the reducer node fails, the data is still accessible after the node is restarted.

\emph{\textbf{Unrecoverable Failures:}} Unrecoverable reducer node failures happen when the reducer node is totally corrupted and not usable at all. In this case, the data in the local store will be lost. To handle this failure, alternative recovery strategies can be adopted for
incremental computation and recomputation.

For incremental computation, recovery is straightforward.
This is because the views from previous jobs cached at the local store
are also stored in the DFS. Thus, when the reducer node is corrupted,
the views can be easily recovered from the DFS.

Under recomputation, the local store caches the sorted intermediate
reduce input data from the merge phase. To handle node failures,
{\tt HaCube} adopts a \textbf{lazy checkpointing} strategy - a snapshot
of the local store is stored to the DFS periodically.
For cube analysis, if we make a snapshot of the cached data after each
view update, it provides the fastest recovery. This is to ensure that data can be directly recovered from the previous view update stage.
However, it is costly to perform checkpointing for each update.

On the other hand, if no snapshots are taken, once a node fails,
we have to recompute it from scratch which is also computationally expensive.
Instead, we advocate an intermediate solution that takes
a snapshot after every $s$ view updates where $s$ can be set
by the users according to the view
update and computer failure frequencies in their cluster.
With such a lazy checkpointing scheme, if a failure happens, the system
can recover by using the most recent snapshot and the new delta
data added after the last checkpointing. Thus, {\tt HaCube} only needs to store the latest snapshot and the data after the snapshot instead of storing all the
base data from the beginning.

\subsection{Storage Cost Discussion}

We argue that {\tt HaCube}'s storage costs are acceptable.

\begin{itemize}
  \item First, even though {\tt HaCube} needs to cache extra data in the local store,
the base data can be deleted from the file system when there is no need to maintain them for other purposes.
The cached data in the local store is sufficient to facilitate view updates
when new data are added.
  \item Second, we can also reduce the number of replicas
stored in the DFS accordingly. The data cached in the local store can
essentially be viewed as one replicated dataset.
  \item Third, {\tt HaCube} only needs to cache one copy of the dataset for different measures in each computation model. Recall that all the measures can be processed together. Thus, for all measures issuing recomputation, the cached sorted raw data is able to serve all of them. For all measures issuing incremental computation, the cached view data can be stored together to reduce storage overhead. For instance, assume that both SUM and MAX need to be calculated, we can store these two views together in the format of $<$dimension attributes, SUM, MAX$>$ instead of maintaining them independently.
\end{itemize}


\section{Performance Evaluation}
\label{sec:experiment}

%
We evaluate {\tt HaCube} on the Longhorn Hadoop cluster in TACC (Texas Advanced Computing Center) \cite{tacc:longhorn}. Each node consists of 2 Intel Nehalem quad-core processors (8 cores) and 48GB memory. By default, the number of nodes used is 35 (and 280 cores).

We perform our studies on the classical dataset generated by TPC-D benchmark generators \cite{tpch}. The TPC-D benchmark offers a rich environment representative of many decision support systems. We study the cube views on the fact table, $lineitem$ in the benchmark.
The attributes $l\_partkey$, $l\_orderkey$, $l\_suppkey$ and $l\_shipdate$ are used as the dimensions and the $l\_quantity$ as the measure. We choose MEDIAN and SUM as the representative functions for evaluation. We report the result based on the average execution time of three runs in each experiment.

\subsection{Cube Materialization Evaluation}

\textbf{Baseline Algorithms} To study the benefit of the optimizations adopted in {\tt CubeGen}, we design two corresponding baseline algorithms to study each of them including {\tt MulR$\_$MulS} (compute each cuboid using one MR job) and {\tt SingR$\_$MulS} (compute all the cuboids using one MR job without batching them), which are widely used for cube computations in MR. {\tt MulR$\_$MulS} is used to study the benefit of removing multiple data read overheads. {\tt SingR $\_$MulS} is used to study the benefit of sharing the shuffle and computation through batch processing.

In the following set of experiments, we vary the data size
from 600M (Million) to 2.4B (Billion) tuples.We study two versions of the {\tt CubeGen} algorithm where {\tt CubeGen} {\tt $\_$Cache} caches the data and {\tt CubeGen$\_$NoCache} does not. This provides insights into the overhead of caching the data to the local store.

\subsubsection{Efficiency Evaluation} We first evaluate the performance improvement of {\tt CubeGen} for cube materialization.
Fig. \ref{fig:excubebuilding} (a) and (b)
show the execution time of all four algorithms for MEDIAN and SUM respectively. As expected, for both MEDIAN and SUM, our
{\tt CubeGen}-based algorithms are 2.2X and 1.6X faster than {\tt MulR$\_$MulS} and {\tt SingR$\_$MulS} on average respectively. This indicates that computing the entire cube in one MR job reduces the overheads significantly compared to the case where multiple MR jobs were issued which requires reading data multiple times. In addition, it also demonstrates that batch processing highly reduces the size of intermediate data which can consequently minimize the overheads of data sorting, shuffling as well as computing.

\begin{figure}[b]
    \centerline{\psfig{file=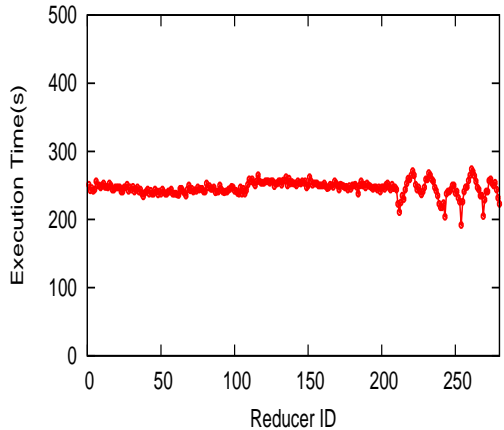,height=6cm, width=7cm} }
    \caption{The load balancing on 280 reducers}
    \label{fig:loadbalancing}
\end{figure}
\subsubsection{Impact of Caching Data} Figure \ref{fig:excubebuilding} (a) and (b) also depict the impact of caching data.
For MEDIAN, we can see that the execution time of the {\tt CubeGen$\_$Cache} algorithm is almost the same as {\tt CubeGen$\_$NoCache} as shown in Fig. \ref{fig:excubebuilding} (a).
This confirms that our optimization to cache the data through file registration
instead of actual data movement does not cause much overhead. For SUM, we observe that {\tt CubeGen$\_$Cache} performs worse than {\tt CubeGen$\_$NoCache}.
This is not surprising as the former needs to write an extra
view to the local file system. However, even though
{\tt CubeGen$\_$Cache} incurs
around 16\% overhead to cache the view,
as we will see later,
it is superior to {\tt CubeGen$\_$NoCache}
when it comes to view updates.

\subsubsection{Load Balancing}

We next show how the $LBCCC$ load balancing scheme works.
The $CCC$ learning job
is conducted using 2 machines and 1GB testing data generated by the benchmark generator.
Then each reducer execution time is recorded to generate a load balancing plan for the {\tt Cube-Gen} algorithm.

We observe that the $LBCCC$ scheme is able to balance the load very well in {\tt CubeGen}.
Fig. \ref{fig:loadbalancing} shows the load situation at each reducer when {\tt CubeGen$\_$NoCache} processes 600M tuples.
We record the {\tt Reduce} phase execution time of each reducer among all
the 280 reducers. We find that 95\% of the reducers complete
their processing within a 10-second difference in execution time.

For the remaining 5\% of the reducers,
as shown in the tail part of the execution time line in
Fig. \ref{fig:loadbalancing}, they take 35 seconds more or less than
the others. This may be caused by the dimension data hash code skew.
We find out that reducers 211 to 280 are assigned to process
the same batch. Recall that within these 69 reducers, the data is
hash partitioned to each reducer.
Thus, if the hash codes of partition attributes are skewed, some reducers will
get more data than others. However, we can see that the average execution
time of these 69 reducers is almost the same as the others which
confirms that our $LBCCC$ does provide each batch the right number of
computation resources. One possible solution to handle
this skew challenge is to adopt the partitioning
mechanisms such as range partitioning to better allocate the
data evenly.

\begin{figure}[h]
    \centerline{\psfig{file=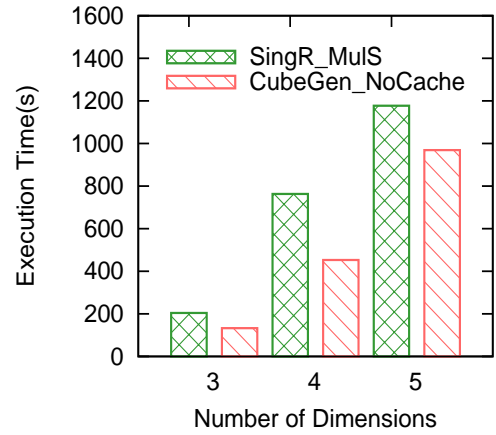,height=6cm, width=7cm} }
    \caption{Impact of Number of Dimensions}
    \label{fig:impactofdimensions}
\end{figure}

\subsubsection{Impact of Number of Dimensions}

We further analyze the impact for cube materialization while varying the number of dimensions from 3 to 5. Our current dataset has 4 dimensions: $l\_partkey$, $l\_orderkey$, $l\_s$ $uppkey$ and $l\_shipdate$. To generate a 3-dimension dataset, we generate the data by removing the $l\_shipdate$ from current dataset. While for the 5-dimension dataset, we generate the data by adding another dimension, $l\_recei- ptdate$ to current dataset.

Figure \ref{fig:impactofdimensions} shows the execution time of {\tt SingR$\_$MulS} and {\tt CubeGen} {\tt $\_$NoCache} for SUM on 600M tuples. Not surprisingly, increasing the number of dimensions increases the cube building time. The results show that {\tt CubeGen$\_$NoCache} outperforms {\tt SingR$\_$MulS} in all these three cases.

\begin{figure*}[t]
\centering
  \begin{minipage}[p]{.55\textwidth}
     \centering
      \epsfig{file=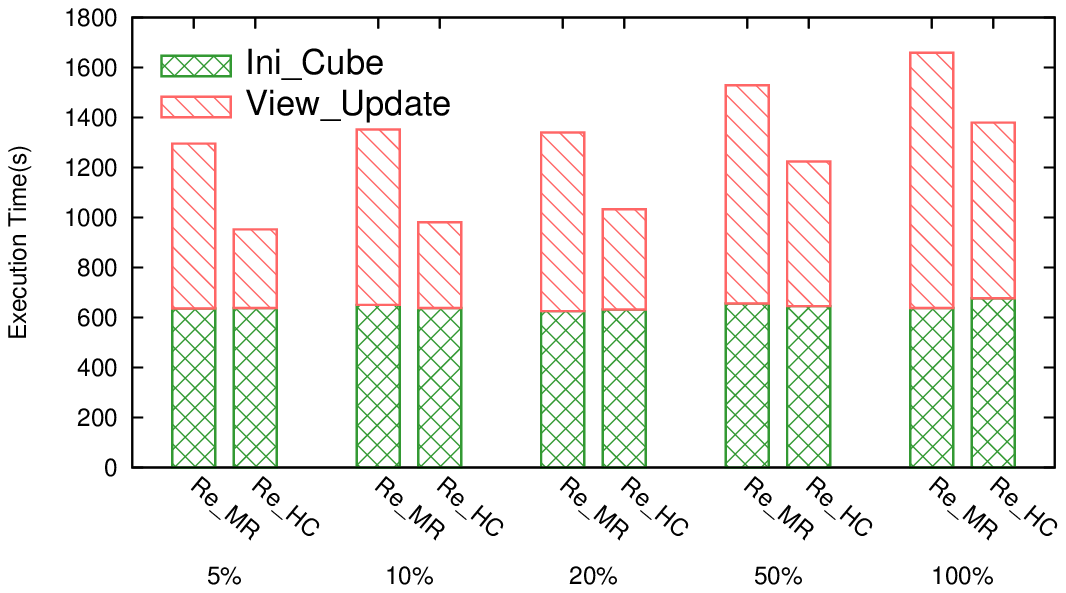,height=5.3cm, width=8cm}\\
      (a) Maintenance: MEDIAN
       \end{minipage}
    \begin{minipage}[p]{.35\textwidth}
   \centering
      \epsfig{file=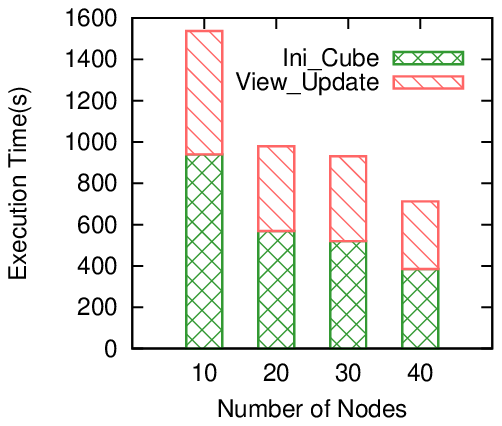,height=5.3cm, width=6cm}\\
      (b) Parallelism: MEDIAN
  \end{minipage}
        \begin{minipage}[t]{.55\textwidth}
   \centering
      \epsfig{file=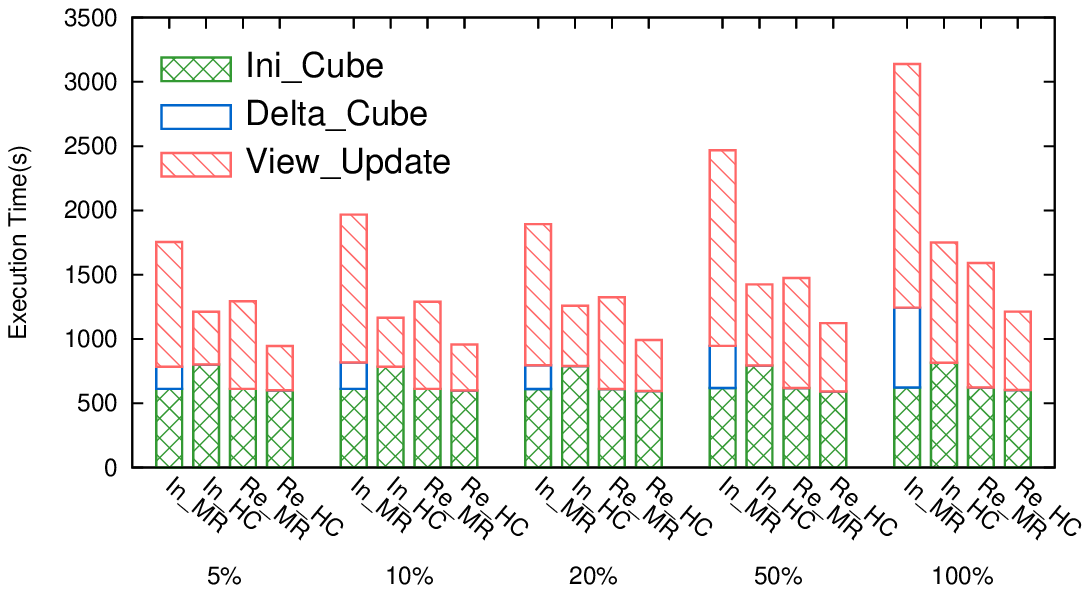,height=5.4cm, width=8cm}\\
      (c) Maintenance: SUM
  \end{minipage}
    \begin{minipage}[t]{.35\textwidth}
   \centering
      \epsfig{file=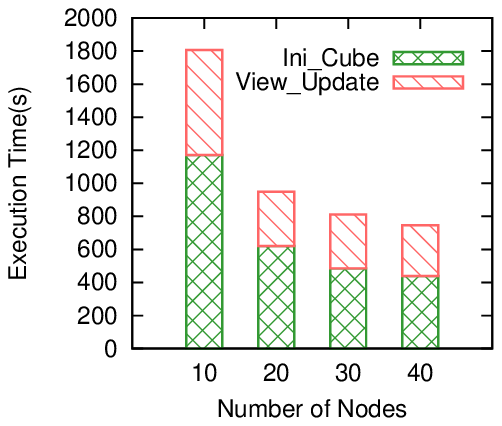,height=5.4cm, width=6cm}\\
      (d) Parallelism: SUM
  \end{minipage}

 \caption{\emph{CubeGen} Performance Evaluation for View Maintenance and Impact of Parallelism}
  \label{fig:viewupdate}
\end{figure*}

\subsection{View Maintenance Evaluation}
\label{subsec:exviewupdate}

\subsubsection{Efficiency Evaluation} We next study the efficiency of performing the view maintenance in {\tt HaCube} compared with Hadoop. We fix $\mathbb{D}$ with 2.4B tuples in the first cube materialization job and vary the size of $\Delta \mathbb{D}$ from 5\% to 100\% of $\mathbb{D}$ for view updates.

Figure \ref{fig:viewupdate} (a) shows the execution time for
both the cube materialization ({\tt Ini$\_$Cube}) and
the view updates ({\tt View$\_$Update}) for MEDIAN.
In this set of experiments, we adopt
recomputation for view updates of MEDIAN using MR ({\tt Re$\_$MR})
and {\tt HaCube} ({\tt Re$\_$HC}).
The result shows that {\tt Re$\_$MR} is 2X and 1.4X faster than {\tt Re$\_$MR}, when $\Delta \mathbb{D}$ is 5\% and 100\% respectively.
The gains come from avoiding reloading and reshuffling $\mathbb{D}$
among the cluster. Thus, the larger $\mathbb{D}$ is, the bigger the
benefit will be.

%

Figure \ref{fig:viewupdate} (c) depicts the result for SUM.
As view updates for SUM can either
be done by incremental computation or recomputation,
we evaluate both approaches.
In Fig. \ref{fig:viewupdate} (c), {\tt In$\_$MR} and {\tt Re$\_$MR} (resp. {\tt In\_HC} and {\tt Re$\_$HC}) are MR (resp. {\tt HaCube}) -based methods using
incremental computation and recomputation respectively.

{\tt In$\_$MR} and {\tt Re$\_$MR} are implemented in the way described in Section \ref{subsec:vminmr}. In {\tt In$\_$MR},
{\tt Delta$\_$Cube} (in the figure) corresponds to
the propagate job to generate the delta view
and {\tt View\_Update} is the refresh job.
The result shows that, for incremental computation, {\tt In\_HC} is 2.8X and 2.2X faster than {\tt In$\_$MR} when $\Delta \mathbb{D}$ is in 5\% and 100\% as shown in Fig. \ref{fig:viewupdate} (c).
For recomputation, {\tt Re$\_$HC} is about 2.1X and 1.4X faster than the {\tt Re$\_$MR} when the $\Delta \mathbb{D}$  is in 5\% and 100\% as shown in Fig. \ref{fig:viewupdate} (c). This indicates that {\tt HaCube} has significant performance improvement compared to MR for the view update for both recomputation and incremental computation.

We observe that incremental computation
performs worse than recomputation in both MR and {\tt HaCube}.
While this seems counter-intuitive, our investigation reveals that
DFS does not provide indexing support; as such, in incremental computation,
the entire view which is much larger than the base data (in our experiments)
has to be accessed. Another insight we gain is the smaller the $\Delta \mathbb{D}$ is, the more effective {\tt HaCube} is.
As future work, we will integrate more existing techniques (e.g. indexing) in DBMS into {\tt HaCube}, which will further improve the view update performance.

%
%

\subsubsection{Impact of Parallelism} We further analyze the impact of parallelism on {\tt HaCube} for both cube materialization and view update while varying the number of nodes from 10 to 40. The experiments use $\mathbb{D}$ with 600M tuples and $\Delta \mathbb{D}$ in 20\% of $\mathbb{D}$ .

Figures \ref{fig:viewupdate} (b) and (d) report the execution time for MEDIAN and SUM. Note that, in this experiment, incremental computation is
used for SUM. We observe that for both recomputation and incremental computation, {\tt HaCube} scales linearly on the testing data set from 10 to 20 nodes, where the execution time almost reduces to half when the resources are doubled. From 20 nodes to 40 nodes, the benefit of parallelism decreases a little bit. This is reasonable, since the entire overheads include two parts, the setup of the framework and the cube computation; the former one may reduce the benefits of increasing the computation resources while cube computation cost is not big enough.

\section{Related Work}
\label{sec:relatedwork}

Much research has been devoted to the problem of data cube analysis \cite{gray:datacube}.
A lot of studies have investigated efficient cube materialization \cite{xin:starcubing}\cite{xin:tkde}\cite{beyer:botoomup}\cite{zhao:arraycube} and view maintenance \cite{lee:incremain}\cite{palpanas:increnondis}. Three classic cube computation approaches (Top-down \cite{zhao:arraycube} , Bottom-Up \cite {beyer:botoomup} and Hybrid \cite{xin:tkde})
have been well studied to share computation among the lattice in a centralized system or a small cluster environment. Different to these approaches, {\tt CubeGen} adopts a new strategy to partition and batch the cuboids according to their prefix order to tackle the new challenges brought by MR. It utilizes the sorting feature better in MR-like systems such that no extra sorting needed during materialization.

%
%

Existing works \cite{sergy:closedcube} \cite{you:closedcube} have
adopted MR to build closed cubes for
algebraic measures. However, both  of these works do not
provide a generic algorithm that can balance the load to
materialize the cube for different measures. Nandi et al. \cite{nandi:mrcube} provided a solution to a special case during the cube computation under MR where one reducer
gets the ``hot spot'' group-by cell with a large number of tuples.
This complements our work and can be employed to handle
such a case when one reducer is overloaded. We note that {\tt HaCube}
is able to support all these existing cube materialization algorithms.
Last but not the least, none of these aforementioned works have
developed any techniques for view maintenance. This is, to the
best of our knowledge, the first work to address the data cube view maintenance
in MR-like systems.

Our work is also related to the problem of incremental computations. Existing works \cite{jorg:clouddb}\cite{bhatotia:incoop}\cite{lammel:mrdeltas} have studied some techniques for incremental computations for single operators in MR.  HaLoop \cite{bu:haloop} is designed to support
iterative operations through a similar caching mechanism which is used for different purposes under a different application context. Restore \cite{elghandour:restore} also shares the similar spirit to keep the intermediate results (either the output of one MR job or the data operated within one job) to DFS in a workflow and reuse them in the future. 
For data cube computation, as the size of intermediate results is large,
{\tt HaCube} adopts a different data caching mechanism to guarantee the data locality that the cached data can be directly used from local store.
This avoids the overhead incurred by Restore in reloading and reshuffling data
from DFS.
Furthermore, none of these existing works provide
explicit support and techniques for data cube analysis
under OLAP and data warehousing semantics.

\section{Conclusion}
\label{sec:conclusion}


It is of critical importance to develop new scalable and efficient data cube analysis systems on a big cluster with low-cost commodity machines to face the challenges brought by the large-scale of data, to provide a better query response and decision making support. In this paper, we made one step towards
developing such a system, {\tt HaCube} an extension of MapReduce, by integrating the good features from both MapReduce (e.g. Scalability) and parallel DBMS (e.g. Local Store). We showed how to batch and share the computations to salvage partial work done by facilitating the features in MapReduce-like systems towards an efficient cube materialization. We also proposed one load balancing strategy such that the load in each reducer can be balanced. Furthermore, we demonstrated how {\tt HaCube} supports an efficient view maintenance by facilitating the extension leveraging a new computation paradigm. The experimental results showed that our proposed cube materialization approach is at least 1.6X to 2.2X faster than the naive algorithms and {\tt HaCube} performs at least 2.2X to 2.8X faster than Hadoop for view maintenance. We expect {\tt HaCube} to further improve the performance by integrating more techniques from DBMS, such as indexing techniques.
%
%
%
%
%
%
%


\bibliographystyle{plain}
\bibliography{CoRR13}  

\end{document}